\documentclass[11pt]{article}
\usepackage{amsmath,amssymb,  amsthm, stmaryrd, MnSymbol, hyperref}
\usepackage{amsfonts}
\usepackage{graphicx}
\usepackage{color}
\usepackage{enumerate}

\newtheorem{theorem}{Theorem}[section]
\newtheorem{fact}[theorem]{Fact}
\newtheorem{corollary}[theorem]{Corollary}
\newtheorem{proposition}[theorem]{Proposition}
\newtheorem{lemma}[theorem]{Lemma}
\newtheorem{remark}[theorem]{Remark}

\theoremstyle{definition}
\newtheorem{definition}[theorem]{Definition}

\begin{document}
\title{\bf  Dynamic Probability Logic: \\
Decidability \& Computability}
\author{Somayeh Chopoghloo {\footnote{\noindent School of Mathematics, Institute for Research in Fundamental Sciences (IPM), PO Box 19395-5746, Tehran, Iran.
			E-mail: s.chopoghloo@ipm.ir.}}
		\hspace{8mm}
		\vspace*{5mm}
		Mahdi Heidarpoor {\footnote{\noindent School of Mathematics, Institute for Research in Fundamental Sciences (IPM), 	PO Box 19395-5746, Tehran, Iran; Department of Mathematics and Computer Science, Amirkabir University of Technology (Tehran Polytechnic), 15194, Tehran, Iran. E-mail: mahdi.heidarpoor@gmail.com.}}
		\hspace{8mm}
		\vspace*{5mm}
		Massoud Pourmahdian {\footnote{\noindent  Corresponding author,
				School of Mathematics, Institute for Research in Fundamental Sciences (IPM),
				PO Box 19395-5746, Tehran, Iran; Department of Mathematics and Computer Science, Amirkabir University of Technology (Tehran Polytechnic), 15194, Tehran, Iran.
				E-mail: pourmahd@ipm.ir.}}	   	
}
\date{}
\maketitle
%%%%%%%%%%%%%%%%%%%%%%%%%%%%%%%%%%%%%%%%%%%%%%%%%%%%%
%%%%%%%%%%%%%%%%%%%%%%%%%%%%%%%%%%%%%%%%%%%%%%%%%%%%%
\begin{abstract}
In this article, the decidability and computability issues of dynamic probability logic ($\mathsf{DPL}$) are addressed. Firstly, a proof system $\mathcal{H}^{-}_{\mathsf{DPL}}$ is introduced for  $\mathsf{DPL}$ and shown that it is weakly complete. Furthermore, this logic has the finite model property and so is decidable.
Secondly, a strongly complete proof system $\mathcal{H}_{\mathsf{DPL}}$ is presented for $\mathsf{DPL}$ and proved that its canonical model is a computable structure.

\end{abstract}
%%%%%%%%%%%%%%%%%%%%%%%%%%%%%%%%%%%%%%%%%%%%%%%%%%%%%
%%%%%%%%%%%%%%%%%%%%%%%%%%%%%%%%%%%%%%%%%%%%%%%%%%%%%
  %Modal %Probability and inductive logic   % Temporal %Logic in computer science      %Discrete-time Markov processes on general state spaces, random dynamical systems
\section{Introduction}
This work is a further continuation of \cite{CP2024} to carry out a research to investigate (propositional modal) dynamic probability logic ($\mathsf{DPL}$) in  which the probability and dynamic (or temporal) logics are combined. The language of $\mathsf{DPL}$ consists of countably many probability modal operators of the form $L_r$ as well as a dynamic (temporal) operator $\bigcirc$.
This logic describes and reasons about dynamic Markov processes which are mathematical structures of the form $\langle \Omega, \mathcal{A}, T, f \rangle$ where $\langle \Omega, \mathcal{A} \rangle$ is a measure space, $T:\Omega\times \mathcal{A}\to [0, 1]$ is a Markov kernel and $f: \Omega\to \Omega$ is a measurable function. Within  $\mathsf{DPL}$, a formula of the form $L_r \varphi$ is interpreted as `the probability of $\varphi$ is at least $r$'. While $\bigcirc\varphi $ means that `the formula $\varphi$ holds in the next instance of time'.
Some strongly complete axiomatic systems of modal probability logic ($\mathsf{PL}$),  were provided earlier in the literature by several authors, e.g. \cite{Hei, Zhou2009, Kozen2}. In \cite{CP2024},  we introduced a Hilbert-style system $\mathcal{H}_{\mathsf{DPL}}$ and proved that it is strongly sound and complete with respect to the class of all dynamic Markov processes.

The first contribution of the current research is to extend the ideas in \cite{Zhou2009} and show that the logic $\mathsf{DPL}$ also enjoys  weak completeness and finite model property as well as decidability. So this logic is computationally tame extension of $\mathsf{PL}$.

Another achievement of this paper is devoted to showing the computable version of strong completeness of $\mathsf{DPL}$. In \cite{CP2024},  to prove the strong completeness,  we introduced a dynamic Markov process $\langle \Omega_c, \mathcal{A}_c, T_c, f_c \rangle$ where $\langle \Omega_c, \mathcal{A}_c \rangle$ is a standard Borel space can be endowed with a valuation function $v_c: \mathbb{P}\to \mathcal{A}_c$ such that any consistent theory $\Gamma$ can be satisfied in $w\in \Omega_c$.
Here using the machinery developed for computable analysis it is proved that the canonical model
$\mathfrak{M}_c=\langle \Omega_c, \mathcal{A}_c, T_c, f_c, v_c \rangle$ is indeed a computable structure.
 The interesting fact is that  when one discard the dynamic part and consider only the probability logic then the canonical model for $\mathsf{PL}$ is also a computable model. So our result also gives a computable version of strong completeness theorem even for $\mathsf{PL}$.

The organization of the paper is as follows. In Section \ref{DPL}, we review the syntax and semantics of $\mathsf{DPL}$ and present an axiomatization of $\mathcal{H}^{-}_{\mathsf{DPL}}$ proving the weak completeness, finite model property and Decidability of  $\mathsf{DPL}$. Section \ref{efsc} is devoted to computability aspects of $\mathsf{DPL}$. In Subsection \ref{SCom}, we recall $\mathcal{H}_{\mathsf{DPL}}$,
introduced \cite{CP2024},  for which the strong soundness and completeness are proved.  In Subsection \ref{ETOP}, we review computability of topological and metric spaces. Subsequently  in Subsection \ref{comcan}, within the framework of computable metric spaces, we show that the canonical model  $\mathfrak{M}_c$ is a computable structure.
%%
%%
%%%%%%%%%%%%%%%%%%%%%%%%%%%%%%%%%%%%%%%%%%%%%%%%%%%%%
%%%%%%%%%%%%%%%%%%%%%%%%%%%%%%%%%%%%%%%%%%%%%%%%%%%%%
\section{Dynamic Probability Logic} \label{DPL}
In this section, we briefly review the syntax and semantics of $\mathsf{DPL}$ and introduce a Hilbert-style axiomatization $\mathcal{H}^{-}_{\mathsf{DPL}}$ proving the weak completeness of this system. Furthermore, we show that $\mathsf{DPL}$ has the finite model property and is decidable.
%%
%%
%%%%%%%%%%%%%%%%%%%%%%%%%%%%%%%%%%%%%%%%%%%%%%%%%%%%%
\subsection{Syntax and semantics} \label{Syntax and semantics}
In this subsection, we review the basic definitions and facts about dynamic probability logic. Most of the concept can be also found in \cite{CP2024}. Suppose  $\mathbb{P}$ is a countable set of propositional variables. The language $\mathcal{L}_\mathsf{DPL}$ of {\em dynamic probability logic} is recursively defined by the following grammar:
\begin{align*}
	\varphi ::&=\ \ p\mid \neg\varphi \mid \varphi \land \varphi \mid L_r \varphi \mid \bigcirc\varphi
\end{align*}
where $p \in \mathbb{P}$ and $r\in \mathbb{Q} \cap [0,1]$. One can use the Boolean connectives $\top$, $\perp$, $\to$, $\leftrightarrow$ and $\vee$ in the usual way. For a formula $\varphi $, the meaning of $L_r \varphi$ is considered as `the probability of $\varphi$ is at least $r$.' Furthermore, $M_r \varphi$ is defined as an abbreviation for $L_{1-r} \neg\varphi$ and interpreted as `the probability of $\varphi $ is at most $r$.' On the other hand, the modality operator $\bigcirc$ is implemented as the {\em dynamic} (or {\em next}) {\em operator} and has dynamic (or temporal) interpretation; $\bigcirc\varphi$ states that `$\varphi$ holds at the next instance of time.'

\begin{definition} \label{Markov}
Let $\langle{\Omega, \mathcal{A}}\rangle$ be a measurable space. A function  $T: \Omega\times \mathcal{A}\to [0, 1]$  is said to be a {\em Markov kernel} (or {\em transition probability}) if it satisfies the following conditions:
\begin{itemize}
\item for each $w \in \Omega$, $T(w, .): A \mapsto T(w, A)$ is a  probability measure on $\mathcal{A}$;
\item  for each $A \in \mathcal{A}$,  $T(., A): w \mapsto T(w, A)$  is a measurable function on $\Omega$.
\end{itemize}
The triple $\langle{\Omega, \mathcal{A}, T}\rangle$ is denoted as a {\em Markov process}  (or {\em Markov chain}) on the state space $\langle{\Omega, \mathcal{A}}\rangle$. In the following,  we may sometimes use $T(w)(A)$ instead of $T(w,A)$.
\end{definition}

\begin{definition}\label{frames}
	 A quadruple $\mathfrak{P} = \langle{\Omega, \mathcal{A}, T, f}\rangle$ is called a {\em dynamic Markov process} %, or
	 whenever $\langle{\Omega, \mathcal{A}, T}\rangle$ is a Markov process and  $f:\Omega\to \Omega$ is a measurable function, and  an element $w\in\Omega$ is called a {\em world} (or {\em state}).
\end{definition}

\begin{definition}
	A {\em dynamic Markov model}, or simply a {\em  model}, is a tuple $ \mathfrak{M} = \langle{\Omega, \mathcal{A}, T, f, v}\rangle$ where $\langle{\Omega, \mathcal{A}, T, f}\rangle$ is a dynamic Markov process  and  $v: \mathbb{P} \to \mathcal{A}$ is a valuation function which assigns to every propositional variable $p \in \mathbb{P}$ a measurable set $v(p)\in\mathcal{A}$.
\end{definition}

The satisfiability relation for arbitrary formulas of $\mathcal{L}_\mathsf{DPL}$ in a given model $ \mathfrak{M} = \langle{\Omega, \mathcal{A}, T, f, v}\rangle$, is
defined inductively as follows:
\begin{itemize}
	\item[] $\mathfrak{M}, w \vDash p \;$ iff $\; w \in v(p)$,
	\item[] $\mathfrak{M}, w \vDash \neg\varphi \;$ iff $\; \mathfrak{M}, w \nvDash \varphi$,
	\item[] $\mathfrak{M}, w \vDash \varphi\land \psi\;$ iff $\; \mathfrak{M}, w \vDash \varphi$ and $\mathfrak{M}, w \vDash \psi$,
	\item[] $\mathfrak{M}, w \vDash L_r \varphi \;$ iff $\; T(w, [\![\varphi]\!]_{\mathfrak{M}}) \geq r$. Here
	$[\![\varphi]\!]_{\mathfrak{M}} = \{ w \in \Omega\;| \; \mathfrak{M}, w \vDash \varphi \}$,
	\item[] $\mathfrak{M}, w \vDash \bigcirc\varphi \;$ iff $\; \mathfrak{M}, f(w) \vDash \varphi$.
\end{itemize}
Moreover, a set $\Gamma$ of formulas holds in the world $w$ if $\mathfrak{M}, w \vDash \varphi$ for each $\varphi \in \Gamma$.
%%
%%

%\begin{fact} \label{continuity}
	%Let $\varphi$ be a formula of $\mathcal{L}_\mathsf{DPL}$, $\mathfrak{M}$ be a model and $r_1, \dots, r_k, r \in \mathbb{Q} \cap [0,1]$. Then,
	%\begin{itemize}
	%\item[1.] The set $[\![L_r \varphi]\!]_{\mathfrak{M}}$ is measurable.
	%\item[2.] Further, the set $[\![L_{r_1\dots r_k} \varphi]\!]_{\mathfrak{M}}$ is measurable, and
	%$$[\![L_{r_1\dots r_k r} \varphi]\!]_{\mathfrak{M}} =\bigcap_{s<r} [\![L_{r_1\dots r_k s} \varphi]\!]_{\mathfrak{M}}.$$
	%\end{itemize}
%\end{fact}
%%
%%

\begin{definition}
    A formula $\varphi$ is {\em valid in a model} $\mathfrak{M}$, denoted by $\mathfrak{M}\vDash \varphi$, if $\mathfrak{M}, w \vDash \varphi$ for all $w \in \Omega$. Likewise, $\varphi$ is \textit{valid in a dynamic Markov process} $\mathfrak{P}$, denoted by $\mathfrak{P}\vDash \varphi$, if it is valid in every model based on $\mathfrak{P}$.
	We say that $\varphi$ is \textit{valid in a class $\mathcal{C}$} of dynamic Markov processes, denoted by ${\vDash}_{\mathcal{C}}\; \varphi$, if it is valid in every element of $\mathcal{C}$.
	Further, $\varphi$ is  a {\em semantic consequence} of a set $\Gamma$ of formulas over a class $\mathcal{C}$  if for every model $\mathfrak{M}$ based on an element in $\mathcal{C}$, if $\Gamma$ holds in a world $w$ in $\mathfrak{M}$, then so is $\varphi$.
	In this case, we write $\Gamma \;{\vDash}_{\mathcal{C}} \;\varphi$.
	For abbreviation, we may omit the subscript $\mathcal{C}$ and instead write ${\vDash}\; \varphi$ and $\Gamma \;{\vDash} \;\varphi$ when it is the class of all dynamic Markov processes.
\end{definition}

\begin{definition}
By {\em dynamic probability logic}, or $\mathsf{DPL}$, we mean the set of all valid formulas of $\mathcal{L}_\mathsf{DPL}$ over the class of all dynamic Markov processes.
\end{definition}
%%
%%
%%%%%%%%%%%%%%%%%%%%%%%%%%%%%%%%%%%%%%%%%%%%%%%%%%%%%
\subsection{Axiomatization} \label{axiomatization}
In this subsection, we provide a (sound) Hilbert-style axiomatization for $\mathsf{DPL}$. We denote this proof system by $\mathcal{H}^{-}_{\mathsf{DPL}}$ .

\begin{definition}
The proof system $\mathcal{H}^{-}_{\mathsf{DPL}}$ has the following axiom schemes:
\begin{align*}
& \textnormal{Taut}&&\text{All propositional tautologies}\\
&\textnormal{FA$_1$}&&L_0 \perp &&\text{(finite additivity)}\\
&\textnormal{FA$_2$}&&L_r \neg\varphi \to \neg L_s\varphi, \text{ if } \; r+s > 1 &&\text{}\\
&\textnormal{FA$_3$}&&L_r (\varphi\land \psi) \land L_s(\varphi\land \neg\psi) \to L_{r+s}\varphi, \text{ if } \; r+s \leq 1&&\text{}\\
&\textnormal{FA$_4$}&& \neg L_r (\varphi\land \psi) \land \neg L_s(\varphi\land \neg\psi) \to \neg L_{r+s}\varphi, \text{ if } \; r+s \leq 1&&\text{}\\
&\textnormal{Mono}&&L_1 (\varphi\to\psi) \to (L_r\varphi\to L_r \psi)  &&\text{(monotonicity)}\\
&\textnormal{Func$_\bigcirc$}&&\bigcirc\neg\varphi\leftrightarrow \neg\bigcirc\varphi &&\text{(functionality)}\\
&\textnormal{Conj$_\bigcirc$}&&\bigcirc(\varphi \wedge \psi) \leftrightarrow (\bigcirc\varphi\wedge\bigcirc\psi) &&\text{(conjunction)}
\end{align*}
and the inference rules:
\begin{align*}
&\textnormal{MP}&&\frac{\displaystyle{\;\varphi\to\psi\;\;\;\varphi\;}}{\displaystyle{\psi}}&& \text{(modus ponens)}\;\\
&\textnormal{Arch$_{\bigcirc^n, r}$}&&\frac{\displaystyle{\;\{\psi\to{\bigcirc}^n L_ s \varphi \;|\; s< r\}\;}}{\displaystyle{\psi \to {\bigcirc}^n L_r \varphi}}, \text{ for } \; n\in \mathbb{N} \;\;\;\;\;\;\;\;\;\;\;\;\;\;\;\;\;\;\;\;\;\;\;\; && \text{(Archimedean)}\\
%& \textnormal{GArch$_{\bigcirc^n, r}$} && \frac{\displaystyle{\;\{\psi\to{\bigcirc}^n L_{r_1\dots r_k s}\varphi \;|\; s< r\}\;}}{\displaystyle{\psi \to {\bigcirc}^n L_{r_1\dots r_k r} \varphi}}, \text{ for }  n\in \mathbb{N} && \text{(generalized Archimedean)} \\
&\textnormal{Nec$_{L_1}$}&&\frac{\displaystyle{\;\varphi\;}}{\displaystyle{L_1 \varphi}} && \text{(necessitation)}\\
&\textnormal{Nec$_\bigcirc$}&&\frac{\displaystyle{{\;\varphi\;}}}{\displaystyle{\bigcirc\varphi}} && \text{}
\end{align*}
\end{definition}

Notice that the Archimedean rule has {\em infinitary} nature in the sense that it allows us to derive a conclusion from a countably infinite set of premises.

\begin{definition}\label{Theorem}
	\begin{itemize}
		\item[1.] A formula $\varphi$ is said to be a {\em theorem} of $\mathcal{H}^{-}_{\mathsf{DPL}}$, or to be {\em derivable} in $\mathcal{H}^{-}_{\mathsf{DPL}}$, (denoted by $\vdash\varphi$) whenever there exists a sequence $\varphi_0, \varphi_1, \dots, \varphi_{\alpha+1}$ ($\alpha$ is a finite or infinite countable ordinal) of the formulas of $\mathcal{L}_\mathsf{DPL}$ with  $\varphi_{\alpha+1} =\varphi$ and for each $\beta \leq{\alpha+1}$,  $\varphi_\beta$ is either an axiom scheme or derived by applying one of the inference rules on some (possibly infinite) preceding formulas of the sequence.
		\item[2.] A formula $\varphi$ is {\em derivable from a set $\Gamma$} of assumptions in $\mathcal{H}^{-}_{\mathsf{DPL}}$ (denoted by $\Gamma  \vdash\varphi$), if there exists a sequence $\varphi_0, \varphi_1, \dots, \varphi_{\alpha+1}$ ($\alpha$ is a finite or infinite countable ordinal) of the formulas of $\mathcal{L}_\mathsf{DPL}$ with $\varphi_{\alpha+1} = \varphi$ and for each $\beta \leq{\alpha+1}$, either $\varphi_\beta$  is a member of $\Gamma$, or a theorem of $\mathcal{H}^{-}_{\mathsf{DPL}}$, or derived by applying one of the inference rules, other than Nec$_{L_1}$ and Nec$_\bigcirc$, on some preceding formulas of the sequence.
		\item[3.] In Parts 1 and 2, the sequence $\varphi_0, \varphi_1, \dots, \varphi_{\alpha+1}$ is called a {\em derivation} (or {\em proof}) {\em of $\varphi$} and a {\em derivation of $\varphi$ from $\Gamma$}, respectively. Further, the length of this sequence is called the {\em length} of that derivation.
	\end{itemize}
\end{definition}

\begin{lemma} \label{prop}
	Let $\varphi$ and $\psi$ be formulas of $\mathcal{L}_{\mathsf{DPL}}$, and $\Gamma$ be a set of formula.
	\begin{itemize}
		\item[1.] $\vdash\!\bigcirc^n (\varphi \to\psi) \leftrightarrow (\bigcirc^n\varphi \to \bigcirc^n\psi)$, for all $n\in \mathbb{N}.$
	    \item[2.] If $\vdash\varphi \to\psi$ then $ \vdash L_r\varphi \to L_r\psi$, for all $r\in \mathbb{Q} \cap [0,1]$.
      %  \item[3.]  $\vdash {\bigcirc}^n L_{r_1\dots r_k r} \varphi \to {\bigcirc}^n L_{r_1\dots r_k s} \varphi$, for all $n\in \mathbb{N}$, $r_1, \dots,  r_k \in \mathbb{Q} \cap [0,1]$ and $s \leq r$.
	\end{itemize}
\end{lemma}
\begin{proof}
Item 1, can be proved by induction on $n$ where the induction base for $n=1$ can be shown by the axioms
Func$_\bigcirc$ and Conj$_\bigcirc$. The second item can be proved using the rule Nec$_{L_1}$ and the axiom Mono.
\end{proof}

Subsequently, the following theorems can be easily shown and therefore the proof is omitted.

\begin{theorem}  [Deduction theorem for $\mathsf{DPL}$] \label{De}
Let $\varphi$ be formula of $\mathcal{L}_{\mathsf{DPL}}$, and $\Gamma$ be a set of formula. Then, $\Gamma, \varphi\vdash \psi$ iff  $\Gamma\vdash\varphi\to \psi$.
\end{theorem}

\begin{theorem}  [Soundness of $\mathcal{H}^{-}_{\mathsf{DPL}}$] \label{sound}
The proof system $\mathcal{H}^{-}_{\mathsf{DPL}}$ is sound with respect to the class $\mathcal{DMP}$ of all dynamic Markov processes, i.e.   if $\vdash\varphi$  then $\vDash \varphi$.
\end{theorem}

As usual, we say that a formula $\varphi$ is {\em consistent} in $\mathcal{H}^{-}_{\mathsf{DPL}}$ if $\nvdash \neg\varphi$. Likewise, a set of formulas $\Gamma$ is  {\em consistent} if $\Gamma\nvdash\perp$.
By deduction theorem the formula $\varphi$ is consistent if and only if  the set $\{\varphi\}$ is so.
Further, $\Gamma$ is {\em finitely consistent} whenever each finite subset of $\Gamma$ is consistent.  Finally,   $\Gamma$ is {\em maximally (finitely) consistent} if it is (finitely) consistent and is not properly included in any (finitely) consistent set of formulas.
%%
%%
%%%%%%%%%%%%%%%%%%%%%%%%%%%%%%%%%%%%%%%%%%%%%%%%%%%%%
\subsection{Weak completeness and decidability} \label{Weakcom}
In this subsection, we verify the weak completeness and decidability of $\mathcal{H}^{-}_{\mathsf{DPL}}$. Some similar results are shown for $\mathsf{PL}$ \cite{Hei, Zhou2009}. Here we use the idea of these results and extended them to  $\mathsf{DPL}$. Indeed, we verify that the methods given by Zhou in \cite{Zhou2009}, also holds for this logic. This is not surprising since the dynamic operator in general
does not have any tight interaction with probability operators that could arise an unpleasant outcome. The following machineries are essential for extending Zhou's methods to the present context.

\begin{definition} \label{depth}
\begin{itemize}
\item[1.]
The {\em probability depth} of a formula $\varphi$, denoted by $p(\varphi)$, is defined inductively as follows:
\begin{align*}
& p(p) := 0 \text{ for any propositional variable } p \in \mathbb{P}; \;\;\;\;\;\;\;\;\;\;\;\;\;\;\;\;\;\;\;\;\;\;\;\;\;\;\;\;\;\;\;\;\;\;\;\;\;\;\;\;\;\;\;\;\;\;\;\;\;\;\;\;\;\;\;\;\;\;\;\;\;\;\;\;\\
& p(\star \varphi)  := p(\varphi) \text{ for } \star\in\{\neg, \bigcirc\};\\
& p(\varphi\land \psi)  := \max{\{pd(\varphi), p(\psi)\}};\\
& p(L_r \varphi)  := p(\varphi) +1.
\end{align*}
\item[2.] Similarly, define the {\em dynamic depth}, $d(\varphi)$ of a formula $\varphi$ as follows:
\begin{align*}
& d(p) := 0 \text{ for any propositional variable } p \in \mathbb{P}; \;\;\;\;\;\;\;\;\;\;\;\;\;\;\;\;\;\;\;\;\;\;\;\;\;\;\;\;\;\;\;\;\;\;\;\;\;\;\;\;\;\;\;\;\;\;\;\;\;\;\;\;\;\;\;\;\;\;\;\;\;\;\;\;\\
& d(\star \varphi)  := d(\varphi) \text{ for } \star\in\{\neg, L_r\};\\
& d(\varphi\land \psi)  := \max{\{d(\varphi),d(\psi)\}};\\
& d(\bigcirc \varphi)  := d(\varphi) +1.
\end{align*}
\end{itemize}
\end{definition}

For a formula $\varphi$ of the form $L_r \theta$, we denote the rational number $r$ as the {\em indice} of formula $\varphi$. For a formula $\psi$ of $\mathcal{L}_\mathsf{DPL}$ and natural numbers $q$, $p$ and $d$, define the {\em local language} $\mathcal{L}(P, q, p, d)$, denoted by $\mathcal{L}[\psi]$, as the largest set of formulas of $\mathcal{L}_\mathsf{DPL}$ that contains $\psi$ and it is closed under the Boolean connectives  and logical equivalence satisfying the following conditions:
\begin{itemize}
	\item[1.] $\mathcal{L}[\psi]$ contains the set $P$ of all propositional variables occurring in $\psi$;
	\item[2.] $q$ is the least common multiple of all denominators of the indices
	of formulas in $\mathcal{L}[\psi]$ and so the indices of all formulas are multiples of $1/q$;
	\item[3.] The formulas in $\mathcal{L}[\psi]$ have the probability and dynamic depth at most $p$ and $d$, respectively.
\end{itemize}

In the above definition, $q$ is called the {\em accuracy} of the language $\mathcal{L}[\psi]$. Put $I[\psi]$ to be the set of all rational numbers of the form of $m/q \in [0,1]$. Call this set as the {\em index set} of $\mathcal{L}[\psi]$. By induction on the probability and dynamic depth of $\psi$, it can be shown that  up to logical equivalence, $\mathcal{L}[\psi]$ is finite.

In the following, a language $\mathcal{L}$ is  meant to be either the whole language $\mathcal{L}_\mathsf{DPL}$ or a finite language of the form $\mathcal{L}(P, q, p, d)$ for some finite set $P$ of propositional variables, and natural numbers $q$, $p$ and $d$. We say that a finite language $\mathcal{L}_1 =\mathcal{L}_1(P_1, q_1, p_1, d_1)$ is {\em more accurate} than the finite language $\mathcal{L}_2 = \mathcal{L}_2(P_2, q_2, p_2, d_2)$ if the accuracy $q_1$ of $\mathcal{L}_1$ is a multiple of the accuracy $q_2$ of $\mathcal{L}_2$.  Moreover, we say that the language $\mathcal{L}_1$ is {\em strictly more accurate} than the language $\mathcal{L}_2$ if $q_1=m\cdot q_2$ for some natural number $m \geq 2$.

We let $\mathcal{L}^+[\psi]$ be the language $\mathcal{L}(P, q, p+1, d)$.
Recall that $w \subseteq \mathcal{L}[\psi]$ is {\em maximally consistent in the language $\mathcal{L}[\psi]$} if it is consistent and no subset of $\mathcal{L}[\psi]$ properly containing $w$ is consistent.  Clearly, any consistent subset $w$ of $\mathcal{L}[\psi]$, is also consistent in $\mathcal{L}^+[\psi]$.
But since $\mathcal{L}^+[\psi]$ is finite, $w$ can be extended to a maximally consistent  subset $w^+$ of $\mathcal{L}^+[\psi]$. Moreover, by a Lindenbaum-like argument,
$w$ can be extended to maximally finitely consistent $w^\infty$ in $\mathcal{L}_\mathsf{DPL}$ extending $w^+$ (see also \cite[Lemma 3.9]{Zhou2009}).

\begin{definition} [Canonical model with respect to $\psi$] A {\em canonical model  with respect to a formula $\psi$} is a tuple $\mathfrak{M}_\psi = \langle{\Omega_\psi, \mathcal{A}_\psi, T_\psi, f_\psi, v_\psi}\rangle$ where
	\begin{itemize}
		\item[1.] $\Omega_\psi$ is  the set  of  all  maximally consistent subsets of $\mathcal{L}[\psi]$;
		\item[2.] $\mathcal{A}_\psi =\{[\varphi]\;|\;\varphi\in \mathcal{L}[\psi]\}$ where $[\varphi]=\{w\in\Omega_\psi\;|\;\varphi\in w \}$;
		\item[3.]  $T_\psi$ is a  function defined for any $w\in \Omega_\psi$  and $\varphi\in \mathcal{L}[\psi]$ by
		\begin{align*}
			{T_\psi}(w)([\varphi])  & =  \sup{\{r\in \mathbb{Q} \cap [0,1] \;|\; L_r\varphi\in w^\infty\}};			
		\end{align*}
		\item[4.] $f_\psi$ is a function defined for any $w\in \Omega_\psi$ by $f_\psi(w) = \{ \varphi\in \mathcal{L}[\psi] \;|\;\bigcirc\varphi\in w\}$;
		\item[5.] $v_\psi$ is a valuation defined by $v_\psi(p) = \{ w\in \Omega_\psi \;|\;p\in w\}$, for each propositional variable $p \in \mathcal{L}[\psi]$.
	\end{itemize}
\end{definition}

Notice that the item 3 above depends on the choice of $w^\infty$, and therefore distinct canonical models might exist. Furthermore, it is not hard to show that
\[{T_\psi}(w)([\varphi]) = \inf{\{r\in \mathbb{Q} \cap [0,1] \;|\; \neg L_r\varphi\in w^\infty\}}.\]

Note that  the finiteness of the set $\mathcal{L}[\psi]$ implies the finiteness of any $w\in \Omega_\psi$.
Moreover,  for every formula $\varphi \in \mathcal{L}[\psi]$, either  $\varphi \in w$ or $\neg\varphi\in w$. Hence, if we take $\varphi_w$ as the conjunction of all formula of $w$, then $\varphi_w\in  \mathcal{L}[\psi]$ and
$[\varphi_w] = \{w\}$. Thus, if $A\subseteq \Omega_\psi$, then $A= [\varphi_A]$ for $\varphi_A = \bigvee_{w\in A} \varphi_w$. So we have  $2^{\Omega_\psi} =\{[\varphi]\;|\;\varphi\in \mathcal{L}[\psi]\}$.

\begin{lemma}[Truth lemma]
 $\mathfrak{M}_\psi$ is a dynamic Markov model. Furthermore, for any formula $\varphi \in \mathcal{L}[\psi]$ and $w \in \Omega_\psi$,
\[\mathfrak{M}_\psi, w\vDash \varphi \text{ iff } \varphi\in w.\]
\end{lemma}
\begin{proof}
The proof is similar to the proof of Lemma 3.18 in \cite{Zhou2009}.
\end{proof}

\begin{theorem}[Weak completeness of $\mathcal{H}^{-}_{\mathsf{DPL}}$] \label{com}
The proof system $\mathcal{H}^{-}_{\mathsf{DPL}}$ is weakly complete with respect to the class of all dynamic Markov processes, i.e. if $\vDash \varphi \text{ then } \vdash \varphi$.
\end{theorem}
\begin{proof}
Suppose that $\nvdash \varphi$. Then $\psi = \neg \varphi$ is consistent in $\mathcal{H}^{-}_{\mathsf{DPL}}$ and so it is contained in some maximally
consistent set $w\subseteq\mathcal{L}[\psi]$. So by truth lemma, $\mathfrak{M}_\psi, w \vDash \psi$. Hence, $\nvDash \varphi$.
\end{proof}

\begin{corollary}[Finite model property] \label{fmp}
If $\varphi\in \mathcal{L}_\mathsf{DPL}$ is satisfiable, then it is satisfiable in a finite model.
\end{corollary}
%%
%%

%decidability
Next, we verify that $\mathcal{H}^{-}_{\mathsf{DPL}}$ is decidable. In the case of probability logic, this is accomplished by translating probability formulas into linear inequalities and using the Fourier-Motzkin elimination method, see \cite[Section 5]{Zhou2009}. In the current situation when the next operator is also present,  using this method is more elaborate. Therefore one should be more careful in defining the linear inequalities corresponding to formulas.

Fix a formula $\psi$ of $\mathcal{L}_\mathsf{DPL}$ and consider $\mathcal{L}[\psi]= \mathcal{L}(P, q, p, d)$ ,    $\mathcal{L}^+[\psi]= \mathcal{L}(P, q, p+1, d)$ and the set of all indices $I[\psi]$  defined as above. One can easily see that any formula in  $ \mathcal{L}^+[\psi]$  is logically equivalent to a disjunction of formulas of the  following form:
\[\bigwedge_k {\bigcirc}^{n_k} p_{i_k} \wedge \bigwedge_{k'} \neg {\bigcirc}^{n_{k'}}  p_{i_{k'}}   \wedge \bigwedge_{l}  {\bigcirc}^{n_{l}} L_{r_{l}} \theta_{j_{l}} \wedge \bigwedge_{l'}  \neg {\bigcirc}^{n_{l'}} L_{r_{l'}} \theta_{j_{l'}}\;\;\;\;(\star) \]
where $ p_{i_k},  p_{i_{k'}}\in P$, $\theta_{j_{l}},  \theta_{j_{l'}} \in \mathcal{L}[\psi]$ and $n_k, n_{k'}, n_{l}, n_{l'}\in \mathbb{N}$, with the convention that $\bigcirc^0 \gamma =\gamma$,  for each formula $\gamma$. Hence, formulas of  the form $(\star)$ can be written as
\[\bigwedge_k {\bigcirc}^{n_k} p_{i_k} \wedge \bigwedge_{k'} \neg {\bigcirc}^{n_{k'}}  p_{i_{k'}}
\wedge \bigwedge_{m}   L_{r_{m}} \theta_{j_{m}} \wedge \bigwedge_{m'}  \neg  L_{r_{m'}} \theta_{j_{m'}}
 \wedge \bigwedge_{l}  {\bigcirc}^{n_{l}} L_{r_{l}} \theta_{j_{l}} \wedge \bigwedge_{l'}  \neg {\bigcirc}^{n_{l'}} L_{r_{l'}} \theta_{j_{l'}} \;\;\;\;(\star\star)\]
where $n_{l}$ and $ n_{l'}$ are positive natural numbers. Now in order to implement the analysis mentioned above, we should eliminate the last two conjuncts from the   formulas of the form $(\star\star)$. To achieve this, we will introduce some new propositional variables.

\begin{definition}
We say that a formula in $\varphi \in \mathcal{L}^+[\psi]$ is in {\em normal form} if it can be written as
\[\bigwedge_k {\bigcirc}^{n_k} p_{i_k} \wedge \bigwedge_{k'} \neg {\bigcirc}^{n_{k'}}  p_{i_{k'}}
\wedge \bigwedge_{m}   L_{r_{m}} \theta_{j_{m}} \wedge \bigwedge_{m'}  \neg  L_{r_{m'}} \theta_{j_{m'}}\]
where $ p_{i_k},  p_{i_{k'}}\in P$, $\theta_{j_{m}},  \theta_{j_{m'}} \in \mathcal{L}[\psi]$ and $n_k, n_{k'}\in \mathbb{N}$.
\end{definition}

\begin{lemma}\label{normal}
For a given formula  $\psi \in \mathcal{L}_{\mathsf{DPL}}$, there exists a finite set $\mathcal{Q}$ of  new propositional variables, depending on the local language $\mathcal{L}[\psi]$, such that each formula $\varphi \in  \mathcal{L}^+[\psi]$ is   logically equivalent to a finite disjunctions of normal  form formulas in
$\mathcal{L}_\mathcal{Q}^+[\psi]= \mathcal{L}^+(P\cup\mathcal{Q}, q, p+1, d)$.
\end{lemma}
\begin{proof}
For $ r\in I[\psi]$ and $\theta\in \mathcal{L}[\psi]$, introduce a new propositional variable $Q_{r, \theta}$ and let
$\mathcal{Q}= \{Q_{r, \theta} |\; r\in I[\psi],\;\theta\in \mathcal{L}[\psi]\}$. We will show our result for this set $\mathcal{Q}$.
To prove the statement, it is sufficient to show the claim for the formulas of the form $(\star\star)$.  It is clear that
any formula of the form  $(\star\star)$ is logically equivalent  to the formula
\begin{align*}
& \bigwedge_k {\bigcirc}^{n_k} p_{i_k} \wedge \bigwedge_{k'} \neg {\bigcirc}^{n_{k'}}  p_{i_{k'}}
\wedge \bigwedge_{m}   L_{r_{m}} \theta_{j_{m}}  \wedge \bigwedge_{l'}  \neg L_{r_{m'}} \theta_{j_{m'}} \\
\wedge &  \bigwedge_l {\bigcirc}^{n_l} Q_{{r_{l}}, \theta_{j_{l}}} \wedge \bigwedge_{l'} \neg {\bigcirc}^{n_{l'}}  Q_{{r_{l'}}, \theta_{j_{l'}}}\\
\wedge & \bigwedge_l (L_{{r_{l}}} \theta_{j_{l}} \leftrightarrow Q_{{r_{l}}, \theta_{j_{l}}}) \wedge \bigwedge_{l'}  (L_{{r_{l'}}} \theta_{j_{l'}} \leftrightarrow Q_{{r_{l'}}, \theta_{j_{l'}}})
\end{align*}
But it is clear that this formula can be written as disjunction of  normal form formulas in $\mathcal{L}_\mathcal{Q}^+[\psi]$, as desired.
\end{proof}

Next, we show the decidability of $\mathcal{H}^{-}_{\mathsf{DPL}}$.
	
\begin{theorem}  \label{decid}
$\mathcal{H}^{-}_{\mathsf{DPL}}$ is decidable.
\end{theorem}
\begin{proof}
	To show the statement, it suffices to prove that the satisfiability problem in $\mathsf{DPL}$ is decidable. But by lemma \ref{normal}, the decidability of whether a formula is satisfiable is reduced to decide the satisfiability of formulas of the normal form. Let $\varphi \in \mathcal{L}^+[\psi]$ is in normal form, i.e.,
\[\varphi = \bigwedge_k {\bigcirc}^{n_k} p_{i_k} \wedge \bigwedge_{k'} \neg {\bigcirc}^{n_{k'}} p_{i_{k'}}
\wedge \bigwedge_{m} L_{r_{m}} \theta_{j_{m}} \wedge \bigwedge_{m'} \neg L_{r_{m'}} \theta_{j_{m'}}\]
where $ p_{i_k}, p_{i_{k'}}\in P$, $\theta_{j_{m}}, \theta_{j_{m'}} \in \mathcal{L}[\psi]$ and $n_k, n_{k'}\in \mathbb{N}$. Put
$\gamma = \bigwedge_k {\bigcirc}^{n_k} p_{i_k} \wedge \bigwedge_{k'} \neg {\bigcirc}^{n_{k'}} p_{i_{k'}}$ and $\gamma' = \bigwedge_{m} L_{r_{m}} \theta_{j_{m}} \wedge \bigwedge_{m'} \neg L_{r_{m'}} \theta_{j_{m'}}$. Since by \cite[Theorem 2.5.6]{Kroger}, it is decidable to check the satisfiability of $\gamma$ as a pure temporal formula, to verify whether $\varphi$ is satisfiable or not, clearly may assume that $\gamma$ is satisfiable. Moreover, $\gamma'$ is satisfiable if and only if $\varphi$ is satisfiable. Now, by the method used in \cite[Section 5]{Zhou2009}, one can effectively associate a finite set $S_{\gamma'}$ of linear system of inequalities with rational coefficients to the formula $\gamma'$ such that $S_{\gamma'}$ is solvable if and only if $\gamma'$ is satisfiable. But, by the well-known result of the decidability of theory of real-closed fields, the solvability of $S_{\gamma'}$ is decidable. Hence, it is decidable whether $\gamma'$ is satisfiable or not.
\end{proof}
%%
%%
%%%%%%%%%%%%%%%%%%%%%%%%%%%%%%%%%%%%%%%%%%%%%%%%%%%%%
%%%%%%%%%%%%%%%%%%%%%%%%%%%%%%%%%%%%%%%%%%%%%%%%%%%%%
\section{Effective strong completeness} \label{efsc}
In Subsection \ref{SCom}, we first present an axiomatization of $\mathcal{H}_{\mathsf{DPL}}$ given previously in \cite{CP2024}, in order to maintain the strong completeness of $\mathsf{DPL}$. Secondly in Subsection \ref{ETOP}, we review the basic definitions and facts about effective topological and metric spaces. Finally, in Subsection \ref{comcan}, we show our main result proving that the canonical model introduced for strong completeness forms a computable structure (Theorem \ref{comstr}).
%%
%%
%%%%%%%%%%%%%%%%%%%%%%%%%%%%%%%%%%%%%%%%%%%%%%%%%%%%%
\subsection{Strong completeness} \label{SCom}
In \cite[Section 2]{CP2024}, an axiomatization $\mathcal{H}_{\mathsf{DPL}}$ is given for $\mathsf{DPL}$ and shown that it is strongly complete with respect to the class of all dynamic Markov processes. This axiomatization extends $\mathcal{H}^{-}_{\mathsf{DPL}}$ through strengthening the Archimedean rule (Arch$_{\bigcirc^n, r}$).

\begin{definition} The proof system $\mathcal{H}_{\mathsf{DPL}}$ consists of all axioms and rules of $\mathcal{H}^{-}_{\mathsf{DPL}}$ (probably) except Arch$_{\bigcirc^n, r}$, plus the following inference rule:
\begin{align*}
\textnormal{GArch$_{\bigcirc^n, r}$} && \frac{\displaystyle{\;\{\psi\to{\bigcirc}^n L_{r_1\dots r_k s}\varphi \;|\; s< r\}\;}}{\displaystyle{\psi \to {\bigcirc}^n L_{r_1\dots r_k r} \varphi}}, \text{ for } \; n\in \mathbb{N} && \;\;\;\;\;\;\;\;\;\;\;\text{(generalized Archimedean)}
\end{align*}
\end{definition}

The notion of the derivation in $\mathcal{H}_{\mathsf{DPL}}$ is defined the same as Definition \ref{Theorem}. Respectively, the other proof theoretic notions such as (finite and maximal) consistency remain the same.
The following results are shown in \cite{CP2024}.

\begin{theorem} \label{soundcom}
The proof system $\mathcal{H}_{\mathsf{DPL}}$ is strongly sound and complete for the class of all dynamic Markov processes, i.e. for each set of formulas $\Gamma$ and formula $\varphi$, we have $\Gamma \vdash\varphi$ iff $\Gamma\vDash\varphi$.
\end{theorem}
\begin{proof}
See \cite[Theorems 2.12 and 2.27]{CP2024}.
\end{proof}

By Theorems \ref{sound}, \ref{com} and \ref{soundcom}, the following corollary is clear.
\begin{corollary}\label{proof}
For any formula $\varphi$ of $\mathcal{L}_\mathsf{DPL}$, $\varphi$ is derivable in $\mathcal{H}^{-}_{\mathsf{DPL}}$ if and only if it is derivable in $\mathcal{H}_{\mathsf{DPL}}$.
\end{corollary}

Notice that since both systems of $\mathcal{H}^{-}_{\mathsf{DPL}}$ and $\mathcal{H}_{\mathsf{DPL}}$ satisfy the deduction theorem, the following corollary is also established.

\begin{corollary}
For any finite set of formulas $\Gamma$ and formula $\varphi$, we have $\varphi$ is derivable from $\Gamma$ in $\mathcal{H}^{-}_{\mathsf{DPL}}$ if and only if it is derivable from $\Gamma$ in $\mathcal{H}_{\mathsf{DPL}}$. So $\Gamma$ is consistent in $\mathcal{H}^{-}_{\mathsf{DPL}}$ if and only if it is consistent in $\mathcal{H}_{\mathsf{DPL}}$.
\end{corollary}

Hereafter, we use the notion of consistency with respect to $\mathcal{H}_{\mathsf{DPL}}$, although by the above corollary the notion of consistency for a finite set of formulas remains the same for both $\mathcal{H}^{-}_{\mathsf{DPL}}$ and $\mathcal{H}_\mathsf{DPL}$. So, in particular, a set $\Gamma$ is finitely consistent in $\mathcal{H}^{-}_{\mathsf{DPL}}$ if and only if it is consistent with respect to $\mathcal{H}_{\mathsf{DPL}}$.

In the rest of this section, we will present the key concepts which constitute the building blocks of the canonical model which is implemented for proving the strong completeness. We will see in the next section that this model is, indeed, a computable dynamic Markov model.

\begin{definition} \label{satu}
A set $w \subseteq \mathbb{F}$ of formulas $\mathcal{L}_\mathsf{DPL}$ is called {\em saturated} if
\begin{itemize}
\item[1.] $w$ is finitely consistent,
\item[2.] $w$ is negation complete, i.e. for every formula $\varphi$ of $\mathcal{L}_\mathsf{DPL}$, either $\varphi \in w$ or $\neg\varphi \in w$, and
\item[3.] $w$ has the Archimedean property, that is, for every formula $\varphi$ of $\mathcal{L}_\mathsf{DPL}$, $n\in \mathbb{N}$ and $r_1, \dots, r_k, r\in \mathbb{Q} \cap [0,1]$, if $\{{\bigcirc}^n L_{r_1\dots r_k s} \varphi\;|\; s<r\}\subseteq w$ then ${\bigcirc}^n L_{r_1\dots r_k r} \varphi \in w$.
\end{itemize}
\end{definition}

We recall that a set of formulas is said to be {\em maximally finitely consistent} if it satisfies Conditions 1 and 2 of the above.

\begin{proposition} Let $w$ be a saturated set of formulas of $\mathcal{L}_\mathsf{DPL}$. Then
\begin{itemize}
\item[1.] $w$ is closed under deduction in $\mathcal{H}_{\mathsf{DPL}}$, that is, for every formula $\varphi$ of $\mathcal{L}_\mathsf{DPL}$, if $\varphi$ is derivable from $w$ in $\mathcal{H}_{\mathsf{DPL}}$ then $\varphi \in w$;
\item[2.] $w$ is consistent in $\mathcal{H}_{\mathsf{DPL}}$.
\end{itemize}
\end{proposition}
\begin{proof}
Part 1 can be proved by the transfinite induction on the length of the derivation of $\varphi$ from $w$. The only interesting case where the derivation ends with an application of the rule GArch$_{\bigcirc^n, r}$. In this case, the third condition of Definition \ref{satu} guarantees that if $\varphi\equiv \theta \to{\bigcirc}^n L_{r_1\dots r_k r} \sigma$ and $w\vdash\varphi$, then $ \theta\to{\bigcirc}^n L_{r_1\dots r_k r} \sigma \in w$.\\
Part 2 immediately follows from Part 1 and this fact that $\perp\notin w$, since $w$ is finitely consistent.
\end{proof}

\begin{corollary}
A set of formulas of $\mathcal{L}_\mathsf{DPL}$ is saturated if and only if it is maximally consistent.
\end{corollary}
\begin{proof}
The direction of right to left is simple, since every maximally consistent set is closed under deduction. The other direction, immediately follows from the above proposition and the negation completeness property of the saturated sets.
\end{proof}

The following lemma is the key component of the proof of strong completeness theorem. Later in Subsection \ref{comcan}, we will give a computable version of this lemma (Lemma \ref{comlin2}).

\begin{lemma} [Lindenbaum lemma] \label{lin}
Let $w$ be a consistent set of formulas of $\mathcal{L}_\mathsf{DPL}$ in $\mathcal{H}_{\mathsf{DPL}}$. Then there exists a saturated (or maximally consistent in $\mathcal{H}_{\mathsf{DPL}}$) set $w^*$ such that $w \subseteq w^*$.
\end{lemma}
\begin{proof}
Let $\varphi_0, \; \varphi_1,\; \varphi_2, \;\dots$ and $s_0, \; s_1,\; s_2, \;\dots$ be two enumerations of all formulas of $\mathcal{L}_\mathsf{DPL}$ and all rational numbers in $[0, 1]$, respectively. We inductively define a sequence
$\Gamma_0 \subseteq \Gamma_1 \subseteq \dots \subseteq\Gamma_k \subseteq \dots$ as follows:
Set $\Gamma_0 := w$ and suppose that $\Gamma_k$ is already defined. Then, we put
\begin{equation*}
\Gamma_{k+1} := \left\{
\begin{array}{rl}
\Gamma_k \cup \{\varphi_k\}\;\;\;\;\; \;\;\;\;\;\;\;\;\;\;\;\;\; \;\;\;\;\;\;\;\;\;\;\;\;\;
& \text{if } \Gamma_k \vdash\varphi_k,\\
\Gamma_k \cup \{\neg\varphi_k\}\;\;\;\;\;\;\;\;\;\;\;\;\;\;\;\; \;\;\;\;\;\;\;\;\;\;\;\;\;\!
& \text{if } \Gamma_k \nvdash\varphi_k \;\text{and $\varphi_k$ is not of}\\
& \text{the form}\;{\bigcirc}^n L_{r_1\dots r_k r} \theta,\\
\Gamma_k \cup \{\neg\varphi_k, \neg {\bigcirc}^n L_{r_1\dots r_k s_l} \theta\}\;\;\;
& \text{if } \Gamma_k \nvdash\varphi_k \;\text{and}\;\varphi_k \text{ is of the form } {\bigcirc}^n L_{r_1\dots r_k r} \theta\\
& \text{and $l$ is the least number such}\\
& \text{that $s_l< r$ and }\; \Gamma_k \nvdash {\bigcirc}^n L_{r_1\dots r_k s_l} \theta.
\end{array} \right.
\end{equation*}
where $n\in \mathbb{N}$ and $r_1, \dots, r_k, r, s_l\in \mathbb{Q} \cap [0,1]$. Notice that such $l$ in the third case always exists. Otherwise if $\Gamma_k\vdash {\bigcirc}^n L_{r_1\dots r_k s_l} \theta$ for all $s_l<r$, then, the rule GArch$_{\bigcirc^n, r}$ implies that $\Gamma_k\vdash {\bigcirc}^n L_{r_1\dots r_k r} \theta$, which is a contradiction.
Now, if we consider
\begin{equation*}
w^* := \bigcup_{k \in \omega}\Gamma_k
\end{equation*}
\\
then it is easy to show that $w^*$ satisfies all the desired properties.
First, by induction on $k$, it can be shown that all $\Gamma_k$ are consistent. All cases are straightforward. The only interesting case is when $\Gamma_{k+1} := \Gamma_k \cup \{\neg\varphi_k, \neg {\bigcirc}^n L_{r_1\dots r_k s_l} \theta\}$ where $\varphi_k \equiv {\bigcirc}^n L_{r_1\dots r_k r} \theta$ and $s_l<r$. Suppose that $\Gamma_{k+1}\vdash\bot$. Then by using the deduction theorem, $\Gamma_{k}\vdash {\bigcirc}^n L_{r_1\dots r_k r} \theta \vee {\bigcirc}^n L_{r_1\dots r_k s_l} \theta$. Now since by Part 2 of Lemma \ref{prop}, we have $\vdash{\bigcirc}^n L_{r_1\dots r_k r} \theta \to {\bigcirc}^n L_{r_1\dots r_k s_l} \theta$,
it follows that $\Gamma_{k}\vdash{\bigcirc}^n L_{r_1\dots r_k s_l} \theta$, which is a contradiction.
Clearly, by the above construction, $w^*$ is negation complete and has the Archimedean property. Therefore, it suffices to show that it is finitely consistent. To see that, let
$w'\subseteq w^*$ be finite and $w'\vdash\perp$. Then, we have $\vdash \neg \bigwedge w'$. Hence, for some $m\in \omega$, $ \neg \bigwedge w'\in \Gamma_m$. On the other hand, $w'\subseteq w^*$ implies that $w' \subseteq\Gamma_{m'}$ for some $m'\in \omega$. So it follows that $\Gamma_k$ is inconsistent for $k= \text{max}\{m, m'\}$, a contradiction.
\end{proof}

Below, we show that the family of saturated sets forms a topological space.

\begin{definition} \label{topo}
Suppose that $\Omega_c$ is the set of all saturated sets of formulas of $\mathcal{L}_\mathsf{DPL}$. Then, let $\tau_c$ be the topology generated by the set $\mathcal{B}_c = \{[\varphi] \;|\;\varphi \in \mathbb{F}\}$ on $\Omega_c$ where for every formula $\varphi\in \mathbb{F}$, $[\varphi] := \{ w\in \Omega_c \;|\; \varphi\in w\}$. In this situation, the pair $\langle{\Omega_c, {\tau}_c}\rangle$ is called the {\em canonical topological space}.
\end{definition}

\begin{remark} \label{Polish}
One can easily check that $\mathcal{B}_c$, introduced in the above definition, forms a basis for $\tau_c$. Let $\Omega^*$ be the set of all maximal finitely consistent sets of formulas of $\mathcal{L}_\mathsf{DPL}$ and $\tau^*$ be the topology generated by the set $\mathcal{B}^* = \{[\varphi]^* \;|\;\varphi \in \mathbb{F}\}$ on $\Omega^*$ where for every formula $\varphi \in \mathbb{F}$, $[\varphi]^* := \{ w\in \Omega^* \;|\; \varphi\in w\}$. By the classical Stone representation theorem, the pair $\langle{\Omega^*, {\tau}^*\!}\rangle$ is a Stone space, i.e., a compact totally disconnected Hausdorff space. Further, it is second countable since $\mathbb{F}$ and hence $\mathcal{B}^*$ are countable. Consequently, $\langle{\Omega^*, {\tau}^*\!}\rangle$ is a Polish space. Note that for every formula $\varphi\in\mathbb{F}$, $[\varphi]= \Omega_c \cap [\varphi]^*$, hence the pair $\langle{\Omega_c, {\tau}_c}\rangle$ is a subspace of $\langle{\Omega^*, {\tau}^*\!}\rangle$.
In fact, it is shown in \cite{Kozen2013} that $\Omega_c$ is a $G_\delta$ subset of $\Omega^*$, and hence by classical theorem of Alexandrov \cite[Theorem 2.2.1]{Sri}, it is a Polish space too.
For more details, see \cite[Section 6]{Kozen2013} and \cite{CP2024}.
\end{remark}

\begin{corollary} \label{A}
Let $\mathcal{A}_c$ be the Borel $\sigma$-algebra generated by $\tau_c$. Then, the pair $\langle{\Omega_c, \mathcal{A}_c}\rangle$ is a standard Borel space.
\end{corollary}
\begin{proof}
It is immediate by Remark \ref{Polish}.
\end{proof}

For a proof of the following fact, see Lemmas 8 to 10 and Theorem 11 in \cite{Kozen2}.

\begin{fact}\label{T}
\begin{itemize}
\item[1.] It is easy to check that $\Omega_c \setminus [\varphi]= [\neg\varphi]$ and $[\varphi] \cap [\psi]=[\varphi \land \psi]$, for every $\varphi, \psi \in \mathbb{F}$. Therefore, the set $\mathcal{B}_c$ is an algebra on $\Omega_c$.
\item[2.] For each $w\in \Omega_c$, the set function $\mu_c(w): \mathcal{B}_c\to [0, 1]$, defined as \[\mu_c(w)([\varphi]):= \sup{\{r\in \mathbb{Q} \cap [0,1] \;|\; L_r\varphi\in w\}} = \inf{\{r\in \mathbb{Q} \cap [0,1] \;|\; \neg L_r\varphi \in w\}},\] is finitely additive and continuous from above at the empty set.
\item[3.] For each $w\in \Omega_c$, $\mu_c(w)$ has an unique extension to a countably additive measure $T_c(w): \mathcal{A}_c\to [0, 1]$. Further, it is a probability measure.
\item[4.] $T_c: \Omega_c\times \mathcal{A}_c\to [0, 1]$ is a measurable function.
\end{itemize}
\end{fact}

\begin{lemma} \label{f}
Consider $f_c: \Omega_c\to 2^{\mathbb{F}}$ defined as ${f_c}(w): = \{ \varphi\in \mathbb{F} \;|\;\bigcirc\varphi \in w\}$. Then,
\begin{itemize}
\item[1.] For each $w \in \Omega_c$, we have $f_c(w) \in \Omega_c$, and hence $f_c$ can be considered as a function from $\Omega_c$ to $\Omega_c$.
\item[2.] $f_c$ is a measurable function.
\end{itemize}
\end{lemma}
\begin{proof}
For Part 1, we have
\begin{itemize}
\item[1.] $f_c(w)$ is finitely consistent. Suppose not. Then, there is a finite subset $w'$ of $f_c(w)$ such that $w' \vdash \perp$. Then, $\vdash \neg \bigwedge w'$.
By the rule Nec$_\bigcirc$, we have $\vdash \bigcirc\neg \bigwedge w'$. So by the axioms Func and Conj$_\bigcirc$, we obtain that $\vdash \neg \bigwedge \bigcirc w'$. So $\bigcirc w' \vdash \perp$. But, this contradicts the finite consistency of $w$, since $\bigcirc w'$ is a finite subset of $w$.
\item[2.] $f_c(w)$ is negation complete. Assume that $\varphi\notin f_c(w)$. Then, we have $\bigcirc\varphi\notin w$. So $\neg\bigcirc\varphi\in w$, since $w$ is negation complete. By the axiom Func$_\bigcirc$, $\bigcirc\neg\varphi\in w$. This implies that $\neg\varphi\in f_c(w)$. Similarly, it can be shown that if $\neg\varphi\in f_c(w)$ then $\varphi\notin f_c(w)$.
\item[3.] $f_c(w)$ has the Archimedean property. Suppose that $\{{\bigcirc}^n L_{r_1\dots r_k s} \varphi\;|\; s<r\}\subseteq f_c(w)$. Then, $\{{\bigcirc}^{n+1} L_{r_1\dots r_k s} \varphi\;|\; s<r\}\subseteq w$ by the definition of $f_c$. Now since $w$ has the Archimedean property, so ${\bigcirc}^{n+1} L_{r_1\dots r_k r} \varphi \in w$. Consequently, ${\bigcirc}^{n} L_{r_1\dots r_k r} \varphi \in f_c(w)$.
\end{itemize}
For Part 2, it suffices to show that for each $\varphi\in\mathbb{F}$, $f_{c}^{-1}([\varphi]) = [\bigcirc\varphi]$.
Let $w\in f_{c}^{-1}([\varphi]) $. Then, we have $f_{c}(w)\in [\varphi]$. Hence, $\varphi\in f_{c}(w)$. This yields that $\bigcirc \varphi\in w$. So, we have $w\in [\bigcirc\varphi]$.
The other direction can be shown similarly.
\end{proof}

In the view of Corollary \ref{A}, Fact \ref{T} and Lemma \ref{f}, we introduce a dynamic Markov model which can be exploited as a tool to establish strong completeness theorem.

\begin{definition}[Canonical model for $\mathsf{DPL}$] \label{cpmodel1}
The \textnormal{canonical model} for $\mathsf{DPL}$ is the tuple $\mathfrak{M}_c = \langle{\Omega_c, \mathcal{A}_c, T_c, f_c, v_c}\rangle$ where
\begin{itemize}
\item[1.] $\langle{\Omega_c, \mathcal{A}_c, T_c}\rangle$ is the Markov process obtained from Fact \ref{T};
\item[2.] $f_c: \Omega_c \to \Omega_c$ defined as ${f_c}(w) = \{ \varphi\in \mathbb{F} \;|\;\bigcirc\varphi \in w\}$ for each $w\in \Omega_c$;
\item[3.] $v_c$ is the valuation defined as $v_c(p) = \{ w\in \Omega_c \;|\;p\in w\}$ for each $p \in \mathbb{P}$.
\end{itemize}
\end{definition}

\begin{lemma}[Truth lemma] \label{truth2}
Consider the canonical model ${\mathfrak{M}}_c$. Then for every formula $\varphi$ of $\mathcal{L}_\mathsf{DPL}$ and $w\in\Omega_c$, we have ${\mathfrak{M}}_c, w \vDash\varphi$ iff $\varphi \in w.$
\end{lemma}
\begin{proof}
Both direction can be simultaneously proved by induction on complexity of a formula $\varphi$. The proofs are straightforward and follow from the satisfiability relation of each logical complexities of formulas.
\end{proof}

The truth lemma yields the strong completeness theorem for $\mathcal{H}_{\mathsf{DPL}}$.

\begin{theorem}[Strong completeness of $\mathcal{H}_{\mathsf{DPL}}$] \label{scompletef}
$\mathcal{H}_{\mathsf{DPL}}$ is strongly complete with respect to the class $\mathcal{MS}$ of all dynamic Markov processes based on standard Borel spaces.
\end{theorem}
\begin{proof}
Let $\Gamma\cup\{\varphi\}$ be a set of formulas of $\mathcal{L}_\mathsf{DPL}$ such that $\Gamma\nvdash\varphi$. Then $\Gamma\cup\{\neg\varphi\}$ is consistent. Therefore, by Lindenbaum lemma, there is a saturated set $w^*$ such that $ \Gamma\cup\{\neg\varphi\}\subseteq w^*$. But by truth lemma, we have ${\mathfrak{M}}_c, w^*\vDash\gamma$ for each $\gamma\in \Gamma$ and ${\mathfrak{M}}_c, w^*\nvDash\varphi$. So we have a dynamic Markov model based on a standard Borel space with a node in which $\Gamma$ is satisfied and $\varphi$ is not satisfied, that is, $\Gamma\nvDash\;\varphi$.
\end{proof}
%%
%%
%%%%%%%%%%%%%%%%%%%%%%%%%%%%%%%%%%%%%%%%%%%%%%%%%%%%%
\subsection{Computability in topological and metric spaces} \label{ETOP}

As said earlier, in this subsection, we collect all necessary definitions and facts about computability of topological spaces in the view of computable analysis (see e.g. \cite{Bra, Hoy, wi, Wi and Gru}). Although, we assume that the reader is familiar with the basic notions of computable analysis, we will sketch the basic definitions of computable analysis needed here. For a more detailed discussion of materials on this subject, one can see the book \cite{wi}.

Recall that a real number $x$ is {\em computable} if there is a computable sequence $(x_i)_{i\in\mathbb{N}}$ of rational numbers which converges effectively to $x$. This means that
there exists a computable function $f: \mathbb{N}\to \mathbb{N}$ such that for each $n$ and $i \geq f(n)$, we have that $|x-x_i|<2^{-n}$. Likewise a sequence $(b_i)_{i\in \mathbb{N}}$ of real numbers is computable if there is a double sequence $(a_{i,k})_{i,k\in \mathbb{N}}$ of computable rational numbers for which the sequence $(a_{i,k})_{k\in \mathbb{N}}$ converges to $b_i$ effectively and uniformly in $i$. In other words, there exists a computable function $g: \mathbb{N}^2 \to \mathbb{N}$ such that for each $i$, $n$ and $k\geq g(i, k)$, we have that $|b_i-a_{i, k}|<2^{-n}$. Notice that the bijective paring function $\langle , \rangle: \mathbb{N}^2 \to \mathbb{N}$ would allow us to consider a double sequence on $\mathbb{N}^2$ as a sequence on $\mathbb{N}$. Therefore, a computable double sequence of real numbers can be considered as a computable sequence of real numbers.

Throughout this subsection, we assume that all topological spaces are Hausdorff and second countable.

\begin{definition}[\cite{Wi and Gru}, Definition 4] \label{comput}
Let $\langle X,\tau\rangle$ be a topological space and $\mathcal{B}$ be a countable base for the topology $\tau$. Suppose that $(O_i)_{i\in \mathbb{N}}$ ia an enumeration of $\mathcal{B}$. The triple $\langle X,\tau, (O_i)_{i\in \mathbb{N}}\rangle$ is a {\em computable topological space} if there is a total computable function $f: \mathbb{N}^3 \to \mathbb{N}$ such that $O_i \cap O_j= \bigcup_{k\in \mathbb{N}} O_{f(i,j, k)}$ for each $i, j\in \mathbb{N}$.
\end{definition}

Similarly, the notion of a computable metric space can be given as follows.

\begin{definition} [\cite{Hoy}, Definition 2.4.1]
A {\em computable metric space} is a triple $\langle X, d, (a_i)_{i\in \mathbb{N}}\rangle$ where $\langle X, d \rangle$ is a metric space and $(a_i)_{i\in \mathbb{N}}$ is a dense subset of $X$ such that the sequence $(d(a_i, a_j))_{i, j\in \mathbb{N}}$ forms a computable double sequence of real numbers.
\end{definition}

Let $\mathcal{B}_d$ be the set consisting of open balls around $a_i$'s with positive rational radius. Then $\mathcal{B}_d$ is a base for the metric topology $\tau_d$. It is known that for a given computable metric space $\langle X, d, (a_i)_{i\in \mathbb{N}}\rangle$, there exists a natural enumeration $(O_i)_{i\in\mathbb{N}}$ of $\mathcal{B}_d$ such that the triple $\langle X, \tau_d, (O_i)_{i\in\mathbb{N}} \rangle$ is a computable metric space \cite[Subsection 2.2]{Hoy}.

\begin{definition} [\cite{Bra}, Definition 5.2.4]
Let $\langle X,\tau, (O_i)_{i\in\mathbb{N}}\rangle$ be a computable topological space. An open subset $U$ of $X$ is called {\em lower-computable} if there is a computable function $f: \mathbb{N} \to \mathbb{N}$ such that $U= \bigcup_{k\in \mathbb{N}} O_{f(k)}$.
\end{definition}

\begin{definition} [\cite{Bra}, Definition 5.2.12] \label{comfun}
Let $\langle X,\tau, (O_i)_{i\in\mathbb{N}}\rangle$ and $\langle Y,\tau', (O'_j)_{j\in\mathbb{N}}\rangle$ be two computable topological spaces. We say that a function $f: X\to Y$ is {\em computable} if $f^{-1}(O'_j)$ is lower-computable uniformly in $j$. This means that there exists a computable function $g:\mathbb{N}^2\to \mathbb{N}$ such that $f^{-1}(O'_j) = \bigcup_{j\in \mathbb{N}} O_{g(j, k)}$.
\end{definition}

Clearly, a computable function is continuous. Now, we turn our attention into the space of probability measures $\mathcal{M}(X)$ of a separable metric space $\langle X,d, (a_i)_{i\in \mathbb{N}}\rangle$. It turns out that the space $\mathcal{M}(X)$ can be endowed with the Prokhorov metric $d_P$. The metric topology provided by this metric coincides with the weak topology on $\mathcal{M}(X)$.

\begin{definition}
Let $\mathcal{M}(X)$ be the set of Borel probability measures over the metric space $X$. This set can be endowed with the {\em weak topology}, which is the finest topology for which $\mu_n \to \mu$ if and only if $\int \! f \mathrm{d}\mu_n \to \int \! f \mathrm{d}\mu$ for all continuous bounded function $f : X \to \mathbb{R}$. This topology is metrizable and if the metric topology of $X$ is Polish, then so is the weak topology on $\mathcal{M}(X)$.
\end{definition}

It is known that if $\mathcal{D} \subset \mathcal{M}(X)$ is the set of those probability measures that are concentrated in finitely many points of $A= \{a_i\; |\; i\in \mathbb{N}\}$
and assign rational values to them, is a dense subset of $\mathcal{M}(X)$ (\cite[Section 6]{Bil}).
The enumerations of $(a_i)_{i\in\mathbb{N}}$ and the set of rational numbers $\mathbb{Q}$ give rise to an enumeration $a_{\mathcal{D}}$ of the set $\mathcal{D}$. $a_{\mathcal{D}} $ consists of measures of the form $\mu_{\langle\langle n_1,\dots , n_k\rangle, \langle m_1,... ,m_k\rangle\rangle}$ that is concentrated over the finite set $\{a_{n_1},\dots , a_{n_k}\}$ and $q_{m_i}$ is the weight of $a_{n_i}$.
\begin{definition} \cite[Section 6]{Bil}
The {\em Prokhorov metric} $\pi$ on $\mathcal{M}(X)$ is defined by:
\[\pi(\mu, \nu) := \inf{\{ \epsilon \in \mathbb{R}^+ \;| \;\mu(A) \leq \nu(A^\epsilon
) + \epsilon \text{ for every Borel set } A \}}\]
where $A^\epsilon = \{x\;| \; d(x, A)< \epsilon\}$.
\end{definition}

\begin{fact}[\cite{Ga}, Proposition B.17] \label{Open}
Assume that $\langle X, d\rangle$ is a matric space and $\nu$ is a measure concentrated on a finite subset $S $ of $X$.
Then, $\pi(\mu, \nu) < \epsilon$ if and only if
$\mu(A^{\epsilon}) > \nu(A) - \epsilon$ for each finite subset $A$ of $S$.
\end{fact}

\begin{fact} [\cite{Hoy}, Proposition 4.1.1] \label{M(X)}
The triple $\langle \mathcal{M}(X), \pi, a_{\mathcal{D}}\rangle$ is a computable metric space.
\end{fact}
%%
%%

%%%%%%%%%%%%%%%%%%%%%%%%%%%%%%%%%%%%%%%%%%%%%%%%%%%%%
\subsection{Computability of canonical model} \label{comcan}

We finally turn our attention into proving the main result of this section showing that the canonical model $\mathfrak{M}_c = \langle{\Omega_c, \mathcal{A}_c, T_c, f_c, v_c}\rangle$ introduced in Definition \ref{cpmodel1} is a computable structure.

In order to achieve to this goal, we have to give an effective enumeration of formulas of $\mathcal{L}_{\mathsf{DPL}}$. So we use standard G\"odel numbering and effectively enumerate formulas of $\mathcal{L}_{\mathsf{DPL}}$. This means that we effectively associate a unique natural number to each symbol of the language $\mathcal{L}_{\mathsf{DPL}}$ and code each formula of $\mathsf{DPL}$ as a finite sequence of symbols in accordance with the usual G\"odel coding methods.
This coding defines a bijective function $\ulcorner . \urcorner$ from the set of all formulas into a computable set $Form \subseteq \mathbb{N}$. This bijection, in particular, gives an enumeration of all formulas according to their G\"odel numbering. So for each $n\in \mathbb{N}$, we let $\varphi_n$ be the unique formula in $\mathsf{DPL}$ whose G\"odel numbering is the $n$-th member of the   set $Form$. Henceforth, in the sequel, suppose that $\mathcal{E}:= (\varphi_0, \varphi_1, \dots, \varphi_n, \dots)$ is the above enumeration.
Decidability of $\mathsf{DPL}$ implies the following straightforward proposition.

\begin{proposition} \label{sets}
The following sets are computable:
\begin{itemize}
\item $T:= \{i\in \mathbb{N}\;|\; \vdash \varphi_i \}$,
\item $C:=\{i\in \mathbb{N}\;|\; \varphi_i \;\text{is consistent}\}$,
\item $D:= \{(i, j) \in \mathbb{N}^2\;|\; \varphi_i \vdash \varphi_j \}$.
\end{itemize}
\end{proposition}

For each $i\in \mathbb{N}$, we let $O_i=[\varphi_i] = \{w\in \Omega_c\;|\; \varphi_i\in w\}$.
The following proposition follows easily from the decidability of $\mathsf{DPL}$.

\begin{proposition}
The triple $\langle \Omega_c, \tau_c, (O_i)_{i\in \mathbb{N}}\rangle$ is a computable topological space.
\end{proposition}
\begin{proof}
To see this, it is enough to prove that there exists a computable function $f: \mathbb{N}^3\to \mathbb{N}$ such that $O_i \cap O_j = \bigcup_{k\in \mathbb{N}} O_{f(i, j, k)}$, for each $i,j\in \mathbb{N}$. But since $\mathcal{B}$ is closed under intersection, if $O_i= [\varphi_i]$ and $O_j= [\varphi_j]$, then $O_i \cap O_j= [\varphi_i\land \varphi_j]$. So define the function $f$ as \[f(i, j, k)= \min{\{l\in\mathbb{N} \; |\;\vdash (\varphi_l\leftrightarrow\varphi_i\land\varphi_j)}\}.\]
By Proposition \ref{sets}, it is clear that the function $f$ is computable.
\end{proof}

In the following, we will show that the space $\langle \Omega_c, \tau_c\rangle$ forms a computable metric space. So we will define a metric $d_c$ which generates the same topology as $\tau_c$.

\begin{definition} Let $w_1, w_2\in \Omega_c$. We say that $w_1$ and $w_2$ {\em disagree} on a formula $\varphi$ if $(\varphi \in w_1 \land \neg \varphi\in w_2)$ or $(\neg\varphi \in w_1 \land \varphi\in w_2)$.
Define
\begin{equation*}
d_c(w_1, w_2) = \left\{
\begin{array}{rl}
0\;\;\; & \text{if } w_1= w_2\\
\dfrac{1}{2^{n_0}} & \text{if } n_0= \min{\{n\in \mathbb{N}\;|\; w_1 \text{ and } w_2 \text{ disagree on }\varphi_n\}}
\end{array} \right.
\end{equation*}
\end{definition}

\begin{proposition}
$d_c$ is a metric on $\Omega_c$. Moreover, the metric topology coincides with $\tau_c$.
\end{proposition}
\begin{proof}
Proving the conditions of metric for $d_c$ is clear. Now, we show that $d_c$ generates the same topology as $\tau_c$. For each $n\in \mathbb{N}$ and $w\in \Omega_c$, let $B_{\frac{1}{2^n}}(w)$ be the ball around $w$ of radius $\frac{1}{2^n}$. Notice that $B_{\frac{1}{2^n}}(w)$ is an open set in $\tau_c$. This is true since if we take $\varphi= \bigwedge_{i=1}^n \theta_i$ where $\theta_i =\varphi_i$ whenever $\varphi_i\in w$ and $\theta_i= \neg\varphi_i$ otherwise, then we have that $w\in [\varphi]\subseteq B_{\frac{1}{2^n}}(w)$. Furthermore, for each $w\in O_n= [\varphi_n]$, it is the case that $w\in B_{\frac{1}{2^n}}(w) \subseteq O_n$. Hence, the metric topology generated by $d_c$ coincides with $\tau_c$.
\end{proof}

\begin{definition} A saturated set $w$ is called {\em computable} if the set $\{n\in \mathbb{N}\;|\; \varphi_n\in w\}$ is computable.
\end{definition}

\begin{remark}
It is not hard to see that a saturated set $w$ is computable if $\{n\in \mathbb{N}\;|\; \varphi_n\in w\}$ is a recursively enumerable set. This is true, since every saturated set is negation-complete.
\end{remark}

\begin{lemma}[Computable Lindenbaum lemma] \label{comlin2}
\begin{itemize}
\item[1.] Let $\varphi$ be a consistent formula of $\mathcal{L}_\mathsf{DPL}$. Then there exists a computable saturated set $w$ such that $\varphi \in w$.
\item[2.] More generally, there is a sequence $(w_n)_{n\in \mathbb{N}}$ of computable saturated sets such that for each consistent formula $\varphi_i$, there exists $k\in \mathbb{N}$ with $\varphi_i\in w_k$. Furthermore, the set $\{(i, k)\in \mathbb{N}^2 \;|\; \varphi_i\in w_k\}$ is computable.
\end{itemize}
\end{lemma}
\begin{proof}
As the first item is a special case of the second one, we only prove the second statement.
By the second item of Proposition \ref{sets}, there exists a computable increasing sequence $i_0< i_1< \dots<i_k<\dots$ of indices of consistent formulas of $\mathcal{L}_{\mathsf{DPL}}$.
So for given $k\in \mathbb{N}$, we define by induction on $l \in \mathbb{N}$ an increasing sequence
$\Gamma_0^k\subseteq \Gamma_1^k \subseteq \dots \subseteq \Gamma_l^k \subseteq \dots$
of finite subsets of formulas as follows:

Set $\Gamma_0^k := \{\varphi_{i_k}\}$ and suppose that $\Gamma_l^k$ is already defined. Then, we put
\begin{equation*}
\Gamma_{l+1}^k := \left\{
\begin{array}{rl}
\Gamma_l^k \cup \{\varphi_l\}\;\;\;\;\; \;\;\;\;\;\;\;\;\;\;\;\;\; \;\;\;\;\;\;\;\;\;\;\;\;\;
& \text{if } \Gamma_l^k \vdash\varphi_l,\\
\Gamma_l^k \cup \{\neg\varphi_l\}\;\;\;\;\;\;\;\;\;\;\;\;\;\;\;\; \;\;\;\;\;\;\;\;\;\;\;\;\;\!
& \text{if } \Gamma_l^k \nvdash\varphi_l \;\text{and $\varphi_l$ is not of}\\
& \text{the form}\;{\bigcirc}^n L_{r_1\dots r_k r} \theta,\\
\Gamma_l^k \cup \{\neg\varphi_l, \neg {\bigcirc}^n L_{r_1\dots r_k s} \theta\}\;\;\;
& \text{if } \Gamma_l^k \nvdash\varphi_l \;\text{and}\;\varphi_l \text{ is of the form } {\bigcirc}^n L_{r_1\dots r_k r} \theta,\\
& \text{and ${\bigcirc}^n L_{r_1\dots r_k s} \theta$ is a formula whose indice is}\\
& \text{the minimum in the list $\mathcal{E}$ and } \Gamma_l^k \nvdash {\bigcirc}^n L_{r_1\dots r_k s} \theta.
\end{array} \right.
\end{equation*}

Now, consider
\begin{equation*}
w_k := \bigcup_{l \in \mathbb{N}}\Gamma_l^k.
\end{equation*}

Similar to the Lindenbaum lemma, one can prove that $w_k$ is a saturated set and $\varphi_{i_k} \in w_k$. Furthermore, one can check effectively and uniformly in $k, l\in\mathbb{N}$ whether $\varphi_l \in\Gamma_{l}^k$ or not. Hence, the set $\{(i, k)\in \mathbb{N}^2\;|\; \varphi_{i}\in w_k\}$ is computable.
\end{proof}

\begin{corollary} \label{cms}
\begin{itemize}
\item[1.] The set $\mathcal{S}$ of all computable saturated sets is dense in $\langle \Omega_c, d_c \rangle$.
\item[2.] Let $(w_n)_{n\in \mathbb{N}}$ be a sequence as given in Lemma \ref{comlin2}. Then, the triple $\langle\Omega_c, d_c, (w_n)_{n\in \mathbb{N}} \rangle$ is a computable metric space.
\end{itemize}
\end{corollary}

In the light of the above corollary and Fact \ref{M(X)}, the following corollary is established.

\begin{corollary} \label{CMS}
Let $\langle\Omega_c, d_c, (w_n)_{n\in \mathbb{N}} \rangle$ be given as Corollary \ref{cms} and $\mathcal{S}_c = \{w_n\;|\;n\in \mathbb{N}\}$. Then, the triple $\langle \mathcal{M}(\Omega_c), \pi, a_{\mathcal{S}_c} \rangle$ is a computable metric space.
\end{corollary}

Now we are ready to prove our main result to show that the canonical dynamic Markov process introduced for the proof of strong completeness theorem \ref{scompletef} is a computable structure. Here computability of this structure means that in addition to $\langle\Omega_c, d_c, (w_n)_{n\in \mathbb{N}} \rangle$ is a computable metric space, the functions $T_c: \Omega_c \times \mathcal{A}_c\to [0, 1]$ and $f_c: \Omega_c \to \Omega_c $
are computable. Notice that we view the function $T_c$ as a mapping from the set $\Omega_c $ into $\mathcal{M}(\Omega_c)$. So the computability of $T_c$ fits into the framework of Definition \ref{comfun}.

\begin{theorem} \label{comstr}
The canonical dynamic Markov process $\mathfrak{P}_c= \langle \Omega_c, \mathcal{A}_c, T_c, f_c \rangle$ is a computable structure.
\end{theorem}
\begin{proof}
Let $d_c$ and $(w_n)_{n\in \mathbb{N}}$ be as defined above. Then, by Corollary \ref{cms}, the triple $\langle \Omega_c, d_c, (w_n)_{n\in \mathbb{N}} \rangle$ is a computable metric space.
Now, we show that $T_c: \Omega_c \to \mathcal{M}(\Omega_c)$ is a computable function.
Let $\nu= \mu_{\langle\langle n_1,\dots , n_k\rangle, \langle m_1,... ,m_k\rangle\rangle} \in a_{\mathcal{D}}$ be a measure that is concentrated over the finite set $S= \{w_{n_1},\dots , w_{n_k}\}$ and $q_{m_i}$ is the weight of $w_{n_i}$. Suppose
$r= \frac{1}{2^n}$ and let $B_r^{\pi}(\nu)$ be the open ball around $\nu$ of radius $r$ with respect to Prokhorov metric.
Then, $T_c^{-1}(B_r^{\pi}(\nu)) =\{w\in \Omega_c\;|\; \pi(T_c(w), \nu) < r\}$. But, in the light of Fact \ref{Open}, the condition $\pi(T_c(w), \nu) < r$ is equivalent to $T_c(w, A^r) > \nu(A) - r$ for each subset $A\subseteq S$.
Let $\mathcal{I}= \{A_1, \dots, A_N\}$ be the set of all non-empty subsets of $S$. Now, we have that $A^r = \bigcup_{w \in A} B_r^{d_c}(w)$ for each $A\in \mathcal{I}$. Since $r= \frac{1}{2^n}$, for each $i \leq N$
and $w\in A_i$, we let $\theta_w= \bigwedge_{j\leq n} \varphi_{j}^c$ where $\varphi_{j}^c = \varphi_{j}$ if $\varphi_{j}\in w$ and $\varphi_{j}^c = \neg \varphi_{j}$ if $\neg\varphi_{j}\in w$.
Put $\psi_i = \bigvee_{w \in A_i} \theta_w$. Then, $A_i^r= [\psi_i]$. Hence,
\begin{align*}
T_c(w, A_i^r) = & T_c(w, [\psi_i])\\
= &\sup{\{s\in \mathbb{Q}\cap[0, 1]\;|\; L_s \psi_i \in w\}}.
\end{align*}
Therefore, $T_c(w, A_i^r) > \nu(A_i) -r$ if and only if there exists $s_i> \nu(A_i)-r$ such that $ L_{s_i} \psi_i\in w$.
Thus,
\[T_c^{-1}(B_r^{\pi}(\nu)) = \bigcap_{i=1}^N \; \bigcup_{s_i> \nu(A_i)- r} [L_{s_i} \psi_i]. \]

It is clear that the set $T_c^{-1}(B_r^{\pi}(\nu))$ is a lower-computable open set uniformly in $\nu$.
Hence, $T_c$ is a computable function.

To prove that the function $f_c$ is computable, notice that $f_c^{-1}([\varphi_i]) = [\bigcirc \varphi_i]$ and therefore, $f_c^{-1}([\varphi_i])$ is a lower-computable open set uniformly in $i$.
\end{proof}

\begin{definition}
The {\em satisfaction function} $Sat_c: \Omega_c\times \mathbb{N} \to \{0, 1\}$ is defined as
\begin{equation*}
Sat_c(w, i) := \left\{
\begin{array}{rl}
1
& \text{if }\mathfrak{M}_c, w\vDash \varphi_i\\
0
& \text{otherwise}.
\end{array} \right.
\end{equation*}
\end{definition}

Now we consider the discrete topology on both sets $\mathbb{N}$ and $\{0, 1\}$, they form of computable metric spaces. Therefore, the following corollary is immediate.

\begin{corollary}
The satisfaction function $Sat_c: \Omega_c\times \mathbb{N} \to \{0, 1\}$ is computable.
\end{corollary}
\begin{proof}
It is enough to show that both sets $Sat^{-1}(0)$ and $Sat^{-1}(1)$ are lower-computable open sets. But,
$$Sat^{-1}(\{1\})=\{(w,i)\;|\; \mathfrak{M}_c, w\vDash\varphi_i \}=\bigcup_{i \in \mathbb{N}}([\varphi_i]\times \{i\}), $$
while $$Sat^{-1}(\{0\}) = \{(w,i)\;|\; \mathfrak{M}_c, w\nvDash\varphi_i \}=\bigcup_{i \in \mathbb{N}}([\neg\varphi_i]\times \{i\}).$$
Therefore, both sets are lower-computable open sets.
\end{proof}
%%
%%
%%%%%%%%%%%%%%%%%%%%%%%%%%%%%%%%%%%%%%%%%%%%%%%%%%%%%
%%%%%%%%%%%%%%%%%%%%%%%%%%%%%%%%%%%%%%%%%%%%%%%%%%%%%
\section{Conclusions and further studies}
In this article, we studied the decidability and computability issues of dynamic probability logic. In Theorem \ref{com}, we proved that the proof system $\mathcal{H}^{-}_{\mathsf{DPL}}$ is weakly complete. We further verified that this system enjoys the finite model property and decidability (Corollary \ref{fmp} and Theorem \ref{decid}). In Section \ref{efsc}, we presented a strongly complete proof system $\mathcal{H}_{\mathsf{DPL}}$ and proved that the canonical model $\mathfrak{M}_c = \langle{\Omega_c, \mathcal{A}_c, T_c, f_c, v_c}\rangle$ is a computable structure (Theorem \ref{comstr}).

In the article  \cite{Zhou2009}, an {\em effective finite model property} for probability logic ($\mathsf{PL}$) is shown, i.e. for a given formula $\varphi \in \mathcal{L}_{\mathsf{PL}}$, there exists an effective upper bound on the size of a model which satisfies $\varphi$. It would be interesting to verify whether the same result for $\mathsf{DPL}$ is true or not. It is also worth studying another extensions of probability logic, such as epistemic probability logic and considering decidability and computability issues as explored here. Another compelling topic for further research is to investigate the computability aspects of infinitary extensions of $\mathsf{DPL}$ introduced in \cite{CP2024}.
%%
%%

%%%%%%%%%%%%%%%%%%%%%%%%%%%%%%%%%%%%%%%%%%%%%%%%%%%%%
%%%%%%%%%%%%%%%%%%%%%%%%%%%%%%%%%%%%%%%%%%%%%%%%%%%%%
%\bibliographystyle{plain}
%\bibliographystyle{alpha}
%\bibliographystyle{unsrt}
%\bibliographystyle{abbrv}
%\bibliography{DPL}
%%%%%%%%%%%%%%%%%%%%%%%%%%%%%%%%%%%%%%%%%%%%%%%%%%%%%
%%%%%%%%%%%%%%%%%%%%%%%%%%%%%%%%%%%%%%%%%%%%%%%%%%%%%

%%
%%
%%%%%%%%%%%%%%%%%%%%%%%%%%%%%%%%%%%%%%%%%%%%%%%%%%%%%%
%%%%%%%%%%%%%%%%%%%%%%%%%%%%%%%%%%%%%%%%%%%%%%%%%%%%%%
\end{document}